\documentclass[journal,draftcls,onecolumn,12pt,twoside]{IEEEtran}
\usepackage{amsfonts,color,morefloats,pslatex}
\usepackage{amssymb,amsthm, amsmath,latexsym}

\newtheorem{theorem}{Theorem}
\newtheorem{lemma}[theorem]{Lemma}

\newtheorem{corollary}[theorem]{Corollary}

\newcommand{\Tr}{{\mathrm{Tr}}}

\newcommand{\gf}{{\mathrm{GF}}}
\newcommand{\PG}{{\mathrm{PG}}}

\newcommand{\cP}{{\mathcal{P}}}
\newcommand{\cB}{{\mathcal{B}}}

\newcommand{\cV}{{\mathcal{V}}}
\newcommand{\cT}{{\mathcal{T}}}
\newcommand{\C}{{\mathcal{C}}}

\newcommand{\cH}{{\mathcal{H}}}
\newcommand{\cR}{{\mathcal{R}}}

\newcommand{\bv}{{\mathbf{v}}}

\newcommand{\bB}{{\mathbf{B}}}

\usepackage{blindtext}

\ifCLASSINFOpdf

\else

\fi

\begin{document}
%
\title{Shortened Linear Codes over Finite Fields\thanks{
Y. Liu was supported by the NSFC (Proj. Nos. 11801414, 11701140)
             and the Natural Science Foundation of Hebei Province of China (Proj. No. A2019210223).
C. Ding was supported by the Hong Kong Research Grants Council,
Proj. No. 16301020.
C. Tang was supported by The National Natural Science Foundation of China (Grant No.
11871058) and China West Normal University (14E013, CXTD2014-4 and the Meritocracy Research
Funds).
}}

\author{
Yang Liu,\thanks{Y. Liu is with College of Mathematics and Information Science, Tianjin University of Commerce, Tianjin, P. R. China.
(email: yangshaohua120@163.com)}
Cunsheng Ding,\thanks{C. Ding is with the Department of Computer Science
                           and Engineering, The Hong Kong University of Science and Technology,
Clear Water Bay, Kowloon, Hong Kong, China (email: cding@ust.hk)}
and
Chunming Tang\thanks{C. Tang is with the School of Mathematics  and Information,
China West Normal University, Nanchong 637002,  China, and also with the Department of Computer Science and Engineering,
 The Hong Kong University
of Science and Technology, Clear Water Bay, Kowloon, Hong Kong (e-mail:
tangchunmingmath@163.com).
}
}

\maketitle

\begin{abstract}
The puncturing and shortening technique are two important approaches to constructing new linear codes from old ones.
In the past 70 years, a lot of progress on the puncturing technique has been made, and many works on
punctured linear codes have been done.  Many families of linear codes with interesting parameters have been obtained with the puncturing technique. However, little research on the shortening technique has been done
and there are only a handful references on shortened linear codes. The first objective of this paper is to prove
some general theory for shortened linear codes. The second objective is to study some shortened codes of
the Hamming codes, Simplex codes, some Reed-Muller codes, and ovoid codes. Eleven families of optimal shortened codes
with interesting parameters are presented in this paper. As a byproduct,  five infinite families of $2$-designs are also constructed from
some of the shortened codes presented in this paper.

\end{abstract}


%
\IEEEpeerreviewmaketitle

\section{Introduction}

An $[n, \kappa, d]$ code over $\gf(q)$ is a $\kappa$-dimensional linear subspace of $\gf(q)^n$
with minimum Hamming distance $d$.  By the parameters of a linear code, we refer to its length,
dimension and minimum distance.
An $[n, \kappa, d]$
code over $\gf(q)$ is called \textit{distance-optimal} (respectively,
\textit{dimension-optimal} and \textit{length-optimal}) if there is no $[n, \kappa, d' \geq d+1]$ (respectively,
$[n, \kappa' \geq \kappa+1, d]$ and $[n' \leq n-1, \kappa, d]$) linear code over $\gf(q)$. An optimal code
is a code that is length-optimal, or dimension-optimal, or distance-optimal, or meets a bound for linear codes.

An important problem in the theory and application of coding theory is the construction
of optimal codes and codes with desirable parameters. To this end, one may construct
a linear code with good or desirable parameters from a known linear code with optimal
or good parameters. There are several standard ways of obtaining linear codes from a
known one. The most famous are the puncturing and shortening methods. Since
extending a linear code increases the code length by one and does not change the code
dimension, the extending method is less interesting in constructing new codes.

Let $\C$ be an $[n, \kappa, d]$ code over $\gf(q)$, and let $T$ be a set of $t$ coordinates in $\C$.
We puncture $\C$ by deleting all the coordinates in $T$ in each codeword of $\C$. The resulting
code is still linear and has length $n-t$, where $t=|T|$. We denote the punctured code by
$\C^T$.
Let $\C(T)$ be the set of codewords which are $0$ on $T$. Then $\C(T)$ is a subcode of
$\C$. We now puncture $\C(T)$ on $T$, and obtain a linear code over $\gf(q)$ with length $n-t$,
which is called a \emph{shortened code}\index{shortened code} of $\C$, and is denoted by $\C_T$.

To explain the motivations of this paper, we recall a general construction of linear codes.
Let $D=\{d_1,d_2,\ldots,d_n\}\subseteq \gf(r)$, where $r=q^{m}$ and $q=p^{s}$ with $p$ a prime.
Let $\Tr_{r/q}$ be the trace function from $\gf(r)$ to $\gf(q)$. Define a linear code of length $n$
over $\gf(q)$ by
\begin{eqnarray}\label{eqn-fundamtheorem}
\mathcal{C}_{D}=\{(\Tr_{r/q}(xd_1),\Tr_{r/q}(xd_2),\ldots,\Tr_{r/q}(xd_n)):x\in \gf(r)\}.
\end{eqnarray}
The set $D$ is called the defining set of $\mathcal{C}_{D}$.
The next theorem says that the defining-set construction in \eqref{eqn-fundamtheorem}
is fundamental \cite{HWW20}.

\begin{theorem}\label{thm-traceconst}
Any linear code of length $n$ over $\gf(q)$ can be expressed as the code
$\mathcal{C}_{D}$ in \eqref{eqn-fundamtheorem}, where $D=\{d_1,d_2,\ldots,d_n\}\subseteq \gf(q^m)$
is a multiset and $m$ is some positive integer.
\end{theorem}

Let $\alpha$ be a primitive element of $\gf(q^m)$. The following code
\begin{eqnarray}\label{eqn-irredcycliccode}
\mathcal{C}(m, q, \alpha)
=\{(\Tr_{q^m/q}(a \alpha^0),  \Tr_{q^m/q}(a \alpha^1), \cdots, \Tr_{q^m/q}(a \alpha^{q^m-2})): a \in \gf(q^m)\}
\end{eqnarray}
is a $[q^m-1, m, (q-1)q^{m-1}]$ cyclic code with check polynomial $M_{\alpha^{-1}}(x)$, which is
the minimal polynomial of $\alpha^{-1}$ over $\gf(q)$ and is
irreducible over $\gf(q)$. The code $\mathcal{C}(m, q, \alpha)$ is said to be
\emph{irreducible.} It is easily seen that the weight enumerator of $\mathcal{C}(m, q, \alpha)$
is $1+(q^m-1)z^{q^{m}-q^{m-1}}$ \cite{DY13}.

The next theorem then follows from Theorem \ref{thm-traceconst} \cite{HengDing19}.

\begin{theorem}\label{thm-motivtheorem}
Every linear code of length $n$ over $\gf(q)$ with dual distance at least $2$ is a
punctured code of an irreducible cyclic code $\mathcal{C}(m, q, \alpha)$ in \eqref{eqn-irredcycliccode}
for some integer $m$ and primitive element $\alpha \in \gf(q^m)$.
\end{theorem}

While cyclic codes form a subclass of linear codes, every linear code with minimum distance at least
two can be punctured from a special irreducible cyclic code $\mathcal{C}(m, q, \alpha)$. This demonstrates
the importance of the very special subclass of irreducible cyclic codes $\mathcal{C}(m, q, \alpha)$ and
the puncturing technique.

In the past ten years, a lot of research on the defining-set construction of linear codes has been done, and
many families of linear codes with interesting parameters have been obtained. Recall that the defining-set
construction is in fact to puncture a code $\mathcal{C}(m, q, \alpha)$. Hence, the puncturing technique
was extensively studied in the past ten years and proved to be an effective approach to obtaining
linear codes with good or desirable parameters. However, every approach has limitations and there is
no exception for the puncturing technique.

Motivated by the success of the
puncturing technique, we ask the following questions on the shortening technique:
\begin{enumerate}
\item Is there a family of cyclic codes such that every linear code with minimum distance at least 2 is a shortened
code of a cyclic code in this family?
\item Is every linear code with minimum distance at least 3 a shortened code of a Hamming code?
\item Can we obtain optimal linear codes or linear codes with desirable parameters by shortening certain known
          families of linear codes or cyclic codes? If the answer is positive,  which known families of linear codes can be shortened for obtaining good linear codes, and
          which coordinate set $T$ leads to a shortened linear code $\C_T$ with good or desirable parameters?
\item Can we develop some general theory for shortened linear codes?
\end{enumerate}

The basic questions above are the major motivations of studying the shortening technique.
In addition to obtaining new linear codes with desirable parameters from old ones, one may have to study
some shortened codes of certain families of linear codes in particular applications. For instance, in order to
obtain 2-designs and Steiner systems, certain shortened codes of a family of ternary codes were investigated in
\cite{TDX19}, where two families of shortened ternary codes with the best-known parameters were obtained.
The study of shortened linear codes in \cite{TDX192} led to a generalization of the Assmus-Mattson theorem.
Shortened and punctured codes were used to prove a generalized MacWilliams Identity in \cite{Goldwasser}.
These important applications of shortened linear codes are additional motivations of this paper.

The last motivation of this paper is that there are only several papers on shortened linear codes
in the literature (see \cite{Chen72}, \cite{Lin72},  \cite{HS73}, \cite{Kasa69},  \cite{NZ15}, \cite{YP}, ).
This is quite amazing, as there are many works on the puncturing technique due to the defining-set
construction.  This fact shows that the shortening technique was overlooked by the coding theory community for 70 years.

It is known that $\C_T=((\C^\perp)^T)^\perp$. This means that in theory shortening a linear code can be obtained by
first performing the dual operation, then the puncturing operation, and finally the dual operation. But this involves
the study of three codes in the sequence $\left( \C^\perp,  (\C^\perp)^T, ((\C^\perp)^T)^\perp \right)$, and is usually
much more complicated. In addition, the puncturing technique also has a limitation in obtaining linear codes with good
parameters in practice.
Consequently, it is still necessary to study the shortening technique.

The objectives of this paper are the following:
\begin{enumerate}
\item Develop some general theory for shortened linear codes.
\item Construct linear codes over $\gf(q)$ with various and good parameters by shortening some known families of
          linear codes over $\gf(q)$.
\end{enumerate}
In this paper, we present eleven families of optimal shortened codes and five families of $2$-designs from some
of the shortened codes.

\section{General results about shortened linear codes}

In this section, we introduce known general results and prove new ones about shortened linear codes.

\subsection{Known general results about shortened linear codes}

Under certain conditions, the dimension of a shortened code $\C_T$ is known and given in the next
theorem  \cite[Theorem 1.5.7]{HP03}.

\begin{theorem}\label{thm-shortenedcode}
Let $\C$ be an $[n, \kappa, d]$ code over $\gf(q)$ and let $T$ be any set of $t$ coordinates. Let $\C^T$ denote the punctured code of $\C$ in all coordinates in
$T$. Then the following hold.
\begin{enumerate}
\item $(\C^\perp)_T=(\C^T)^\perp$ and $(\C^\perp)^T =(\C_T)^\perp$.
\item If $t<d$, then $\C^T$ and $(\C^\perp)_T$ have dimensions $\kappa$ and
      $n-t-\kappa$, respectively.
\item If $t=d$ and $T$ is the set of coordinates where a minimum weight codeword
      is nonzero, then $\C^T$ and $(\C^\perp)_T$ have dimensions $\kappa-1$ and
      $n-t-\kappa+1$, respectively.
\end{enumerate}
\end{theorem}

We will use this theorem to settle the dimension for the case $t < d^\perp$, and will develop a case-specific
method to determine the dimension of $\C_T$ for the case $t > d^\perp$ according to the specific design of the
original code $\C$. The most difficult task is to determine the minimum distance $d(\C_T)$ of a shortened code $\C_T$.

Let $\C$ be an $[n, \kappa, d]$ code over $\gf(q)$. Let $M$ be the $q^\kappa \times n$ matrix whose rows are all codewords
in $\C$, and let $M_i$ be the submatrix of $M$ consisting of the codewords of weight $i$. A code is {\em homogeneous}
provided that for $0 \leq i \leq n$, each column of $M_i$ has the same weight. Prange proved the following result
 \cite[p. 271]{HP03}.

\begin{theorem}\label{thm-Prange}
Let $\C$ be an $[n, \kappa, d]$ code over $\gf(q)$ with $d>1$, and let $\C^*$ and $\C_*$ be the code obtained from $\C$
by puncturing and  shortening on some coordinate, respectively. Then for $0 \leq i \leq n-1$, we have
$$
A_i(\C^*)=\frac{n-i}{n}A_i(\C) + \frac{i+1}{n}A_{i+1}(\C)
$$
and
$$
A_i(\C_*)=\frac{n-i}{n} A_i(\C),
$$
where $A_i(\C)$ denotes the number of codewords of weight $i$ in $\C$.
\end{theorem}

It is known that $\C$ is homogeneous if $\C$ has a transitive automorphism group  \cite[p. 271]{HP03}.
This is about the only way to decide if a code is homogeneous.  Hence, Theorem  \ref{thm-Prange} has
very limited applicability.

Recently, better results regarding the weight distribution of a shortened code $\C_T$ of special linear codes
were developed in \cite{TDX192}. To introduce them, we need to introduce combinatorial $t$-designs.

Let $\mathcal P$ be a set of $n$ elements
and $\mathcal B$ a multiset of $b$ $k$-subsets of $\mathcal P$, where $n \ge 1$,
$b \ge 0$ and  $1 \le k \le n$.  Let $t$ be a positive integer satisfying $1 \le t \le n$.
The pair $\mathbb  D=(\mathcal P, \mathcal B)$ is called a  $t$-$(n, k, \lambda )$ \emph{design}, or simply \emph{$t$-design},
if  every $t$-subset of $\mathcal P$ is contained in exactly $\lambda$ elements of $\mathcal B$.  The elements of $\mathcal P$ are called \emph{points},
and those of $\mathcal B$ are referred to as \emph{blocks}.

When $\mathcal B=\emptyset$, i.e., $b=0$, we put $\lambda=0$ and call
$(\mathcal P, \emptyset)$ a $t$-$(n, k, 0)$ design for any $t$ and $k$
with $1 \leq t \leq n$ and $0 \leq k \leq n$. A $t$-$(n, k, \lambda)$
design with $t > k$ must have $\lambda=0$ and must be the design
$(\mathcal P, \emptyset)$. These designs are called trivial designs.
We will use the following conventions for the ease of description in the sequel.
 A $t$-$(n, k, \lambda)$ design $(\mathcal P, \mathcal B)$
is also said to be trivial if every $k$-subset of $\mathcal P$ is a block.

A $t$-design is called \emph{simple} if $\mathcal B$ does not contain repeated blocks.
 A $t$-$(n, k, \lambda)$ design is called a \emph{Steiner system} and denoted by
$S(t,k,n)$ if $t \ge 2$ and $\lambda=1$. The parameters of a $t$-$(n,k, \lambda )$ design satisfy:
\begin{align*}
\binom{n}{t} \lambda =\binom{k}{t} b.
\end{align*}

Let $\mathrm{GF}(q)$ denote the finite field with $q$ elements, where
$q$ is a prime power.  We assume that the reader is familiar with the basics of linear codes.
Let $\mathcal C$ be an $[n, k, d]$ linear code over $\mathrm{GF}(q)$. Let $A_i:=A_i(\mathcal C)$, which denotes the
number of codewords with Hamming weight $i$ in $\mathcal C$, where $0 \leq i \leq n$. The sequence
$(A_0, A_1, \cdots, A_{n})$ is
called the \textit{weight distribution} of $\mathcal C$, and $\sum_{i=0}^{n} A_i z^i$ is referred to as
the \textit{weight enumerator} of $\mathcal C$.
Then the $q$-ary linear code $\mathcal C$  may induce a $t$-design under certain conditions, which is formed by
the supports of  codewords of a fixed Hamming weight in $\mathcal C$.
 Let
$\mathcal P(\mathcal C)=\{0,1, \dots, n-1\}$ be the set of the coordinate positions of $\mathcal C$, where $n$ is the length of $\mathcal C$.
For a codeword $\mathbf c =(c_0, \dots, c_{n-1})$ in $\mathcal C$, the \emph{support} of  $\mathbf c$
is defined by
\begin{align*}
\mathrm{Supp}(\mathbf c) = \{i: c_i \neq 0, i \in \mathcal P(\mathcal C)\}.
\end{align*}
Let $\mathcal B_{w}(\mathcal C)
=\frac{1}{q-1}\{\{   \mathrm{Supp}(\mathbf c): wt(\mathbf{c})=w
~\text{and}~\mathbf{c}\in \mathcal{C}\}\}$, here and hereafter $\{\{\}\}$ is the multiset notation and
$\frac{1}{q-1}S$
denotes the multiset obtained after dividing the multiplicity of each element in the multiset $S$
by $q-1$.  For some special $\mathcal C$, $\left (\mathcal P(\mathcal C),  \mathcal B_{w}(\mathcal C) \right)$
is a $t$-$(n, w, \lambda)$ design with $b$ blocks, where
\begin{align}\label{eq:rl-Ac-d}
b=\frac{1}{q-1} A_{w}, \
\lambda= \frac{\binom{w}{t}}{ (q-1)\binom{n}{t} } A_{w}.
\end{align}
If  $\left (\mathcal P(\mathcal C),  \mathcal B_{w}(\mathcal C) \right)$ is a $t$-design for any $0\le w \le n$,
we say that the code $\mathcal C$ \emph{supports $t$-designs}. Notice that such a design
$\left (\mathcal P(\mathcal C),  \mathcal B_{w}(\mathcal C) \right)$ may have repeated
blocks or may be simple or trivial.

The following lemma provides a criterion for guaranteeing  a simple block set $\mathcal B_k(\mathcal C)$  \cite[Lemma 4.1]{DingBk2}.

\begin{lemma}\label{lem-simpled}
Let $\mathcal C$ be a linear code over $\mathrm{GF}(q)$ with length $n$ and minimum weight $d$. Let $w$ be the largest integer with $w \le n$ satisfying
\begin{align*}
w-\left \lfloor \frac{w+q-2}{q-1} \right  \rfloor <d.
\end{align*}
Then there are no repeated blocks in $\mathcal B_k(\mathcal C)$ for any $d\le k \le w$. Such a block set is
said to be simple.
\end{lemma}

The following theorem gives a characterization of codes supporting $t$-designs via the weight distributions of their shortened and punctured codes \cite{TDX192}, and will be employed later in this paper.

\begin{theorem}\label{thm-tdesign-wtcode}
Let $\mathcal C$ be an $[n, m, d]$ linear code over $\mathrm{GF}(q)$ and  $d^{\perp}$  the minimum distance of $\mathcal  C^{\perp}$.
 Let  $t$ be a positive integer with  $0< t <\min \{d, d^{\perp}\}$. Then the following statements are equivalent.

 (1) $\left ( \mathcal P(\mathcal C) , \mathcal B_k(\mathcal C) \right )$ is a $t$-design for any $0 \le k \le n$.

 (2)  $\left ( \mathcal P(\mathcal C^{\perp}) , \mathcal B_k(\mathcal C^{\perp}) \right )$ is a $t$-design for any $0\le k \le n$.

 (3) For any $1 \le t' \le t$, the weight distribution $\left ( A_k(\mathcal C_T) \right )_{k=0}^{n -t'}$ of the shortened code $\mathcal C_T$
 is independent of the specific choice of the elements in $T$,
 where $T$ is  any  set of $t'$ coordinate positions in $\mathcal P\left ( \mathcal   C \right )$.

  (4) For any $1 \le t' \le t$, the weight distribution $\left ( A_k(\mathcal C^T) \right )_{k=0}^{n-t'}$ of the punctured code $\mathcal C^T$
  is independent of the specific choice of the elements in $T$,
 where $T$ is  any  set of $t'$ coordinate positions in $\mathcal P\left ( \mathcal   C \right )$.
\end{theorem}

Recall that the binomial coefficient $\binom{a}{b}$ equals $0$ when $a<b$ or $b<0$. We have the following useful
result \cite{TDX192}.

\begin{theorem}\label{thm-sctcode}
Let $\mathcal C$ be an $[n, m, d]$ linear code  over $\mathrm{GF}(q)$ and  $d^{\perp}$  the minimum distance of $\mathcal  C^{\perp}$.
 Let  $t$ be a positive integer with  $0< t <\min \{d, d^{\perp}\}$.
Let $T$ be  a  set of $t$ coordinate positions in $\mathcal P\left ( \mathcal   C \right )$.
Suppose that $\left ( \mathcal P(\mathcal C) , \mathcal B_i(\mathcal C) \right )$ is a $t$-design for any $i$ with $d \le i \le n-t$.
Then the shortened code $\mathcal C_T$ is a linear code of length $n-t$ and dimension $m-t$. The weight distribution
$\left ( A_k(\mathcal C_T) \right )_{k=0}^{n-t}$ of $\mathcal C_T$ is independent of the specific choice of the elements
in $T$. Specifically,
$$A_k(\mathcal C_T) =\frac{ \binom{k}{t} \binom{n-t}{k}}{ \binom{n }{t} \binom{n-t}{k-t}}A_k(\mathcal C).$$
\end{theorem}

Theorem \ref{thm-sctcode} is useful, and will be employed to determine the weight distributions of some shortened codes of
several families of linear codes later. The following theorem will also be used later  \cite{TDX192}.

\begin{theorem}\label{thm-pctcode}
Let $\mathcal C$ be an $[n, m, d]$ linear code over $\mathrm{GF}(q)$ and  $d^{\perp}$  the minimum distance of $\mathcal  C^{\perp}$.
 Let  $t$ be a positive integer with  $0< t <d^{\perp}$.
Let $T$ be  a  set of $t$ coordinate positions in $\mathcal P\left ( \mathcal   C \right )$.
Suppose that $\left ( \mathcal P(\mathcal C) , \mathcal B_i(\mathcal C) \right )$ is a $t$-design for any $i$ with $d \le i \le n$.
Then the punctured code $\mathcal C^T$ is a linear code of length $n-t$ and dimension $m$. The weight distribution
$\left ( A_k(\mathcal C^T) \right )_{k=0}^{n-t}$ of $\mathcal C^T$ is independent of the specific choice of the elements
in $T$. Specifically,
$$A_k(\mathcal C^T) =\sum_{i=0}^t  \frac{\binom{n-t}{k} \binom{k+i}{t} \binom{t}{i} }{\binom{n-t}{k-t+i} \binom{n}{t}} A_{k+i}(\mathcal C).$$
\end{theorem}

\subsection{Some new general results}

In this section, we prove two general results. The first one is the following.

\begin{theorem}\label{thm-maingen1}
Every linear code $\C$ over $\gf(q)$ with minimum distance $d \geq 2$ is a shortened code of $\C(m,q,\alpha)^\perp$
for some $m$, $q$ and $\alpha$, where $\alpha$ is a generator of $\gf(q)^*$ and $\C(m,q,\alpha)$ was defined in
(\ref{eqn-irredcycliccode}).
\end{theorem}

\begin{proof}
By assumption $d \geq 2$. It follows from Theorem \ref{thm-motivtheorem} that there are $m$, $q$ and a generator of
$\gf(q)^*$ such that
$$
\C^\perp = \C(m, q, \alpha)^T,
$$
where $T$ is a set of coordinates in $\C(m, q, \alpha)$. It then follows from Theorem \ref{thm-shortenedcode} that
$$
\C^\perp = \C(m, q, \alpha)^T= ( (\C(m, q, \alpha)^\perp )^\perp )^T = ( (\C(m, q, \alpha)^\perp )_T )^\perp.
$$
Hence, $\C=(\C(m, q, \alpha)^\perp )_T$. This completes the proof.
\end{proof}

The following is a corollary of Theorem \ref{thm-motivtheorem}.

\begin{corollary}\label{cor-march29}
Every linear code over $\gf(q)$ with dual distance at least $3$ is a punctured code of a Simplex code
over $\gf(q)$.
\end{corollary}

Note that the dual of a Hamming code is called a Simplex code, which is a one-weight code.

\begin{theorem}\label{thm-maingen2}
Every linear code with minimum distance at least $3$ is a shortened code of a Hamming code over $\gf(q)$.
\end{theorem}

\begin{proof}
The desired conclusion follows from Corollary \ref{cor-march29} and Theorem \ref{thm-shortenedcode}.
The proof is similar to that of  Theorem \ref{thm-maingen1} and is omitted here.
\end{proof}

We now prove the following result.

\begin{theorem}\label{thm-april4101}
Let $\C$ be an $[n, k, d]$ code over $\gf(q)$, and let $d^\perp$ denote the minimum distance
of the dual code $\C^\perp$. Let $t$ be an integer with $1 \leq t < \min\{d, d^\perp\}$. For any set
$T=\{i_1, i_2, \ldots, i_t\}$ of $t$ coordinates,  $\C^T$ has length $n-t$, dimension $k$ and minimum
distance at least $d-t$.

Furthermore, if $A_d(\C) > q^k - q^{k-t}(q-1)^t-1$ and $t \leq k$, then the minimum distance of
$\C^T$ equals $d-t$.
\end{theorem}

\begin{proof}
Since $t < d$, the punctured versions of any two distinct codewords in $\C$ are
distinct. It then follows that the dimension of $\C^T$ is $k$. Clearly,
the minimum distance of $\C^T$ is at least $d-t$.

Let $M$ be the $q^k \times n$ matrix whose rows are all codewords of $\C$. It is
well known that $M$ is an orthogonal array of strength $t$. Let $S_1(T)$ denote
the set of all codewords in $\C$ whose coordinates in $T$ are all nonzero, and
define $S_2(T)=\C \setminus S_1(T)$. By definition, $S_1(T)$ and $S_2(T)$ partition
$\C$. Since $M$ is an orthogonal array of strength $t$, $|S_1(T)|=q^{k-t}(q-1)^t$.
Consequently,
$
|S_2(T)|=q^k-q^{k-t}(q-1)^t.
$
Let $C_d$ denote the set of all codewords of weight $d$ in $\C$. Then $S_1(T) \cap C_d$
and  $S_2(T) \cap C_d$ partition $C_d$. Note that $d\geq 1$ and
$$
|S_2(T) \cap C_d| \leq |S_2(T)|-1 = q^k - q^{k-t}(q-1)^t-1.
$$
If $A_d(\C) > q^k - q^{k-t}(q-1)^t-1$ and $t \leq k$, then $|S_1(T) \cap C_d|\geq 1$.
This means that
there is at least one codeword
in $\C$ whose coordinates in $T$ are all nonzero. As a result, the punctured
version of this codeword has Hamming weight $d-t$. Hence, the minimum distance of
$\C^T$ equals $d-t$.
\end{proof}

The only new result in Theorem \ref{thm-april4101} is the last conclusion on the minimum distance
of the punctured code $\C^T$. The following theorem shows that the minimum distance of $(\C_T)^\perp$
can be determined in some cases.

\begin{theorem}\label{thm-april4102}
Let $\C$ be an $[n, k, d]$ code over $\gf(q)$, and let $d^\perp$ denote the minimum distance
of the dual code $\C^\perp$. Let $t$ be an integer with $1 \leq t < \min\{d, d^\perp\}$. For any set
$T=\{i_1, i_2, \ldots, i_t\}$ of $t$ coordinates,  $(\C_T)^\perp$ has length $n-t$, dimension $n-k$ and minimum
distance at least $d^\perp-t$.

Furthermore, if $A_{d^\perp}(\C^\perp) > q^{n-k} - q^{n-k-t}(q-1)^t-1$ and $t \leq n-k$, then the minimum distance of
$(\C_T)^\perp$ equals $d^\perp-t$.
\end{theorem}

\begin{proof}
By Theorem \ref{thm-shortenedcode}, we have $(\C_T)^\perp=(\C^\perp)^T$. The desired conclusions then follow
from Theorem \ref{thm-april4101}.
\end{proof}

Both Theorems \ref{thm-april4101} and \ref{thm-april4102} will be used to determine the parameters of some
shortened codes and their duals later.

\section{Some shortened codes of the Hamming codes}

A parity check matrix $H_{(q,m)}$ of the \emph{Hamming code\index{Hamming code}} $\cH_{(q,m)}$ over $\gf(q)$
 is defined by choosing for its columns a nonzero vector from
each one-dimensional subspace of $\gf(q)^m$. In terms of finite geometry, the columns of $H_{(q,m)}$ are the
points of the projective geometry $\PG(m-1, \gf(q))$. Hence  $\cH_{(q,m)}$ has length $n=(q^m-1)/(q-1)$ and dimension
$n-m$. Note that no two columns of $H_{(q,m)}$ are linearly dependent over $\gf(q)$. The minimum weight of $\cH_{(q,m)}$
is at least 3. Adding two nonzero vectors from two different one-dimensional subspaces gives a nonzero vector
from a third one-dimensional space. Therefore,  $\cH_{(q,m)}$ has minimum weight 3.
It is also well known that any $[(q^m-1)/(q-1), (q^m-1)/(q-1)-m, 3]$ code over $\gf(q)$ is monomially equivalent to the Hamming code $\cH_{(q,m)}$ \cite[Theorem 1.8.2]{HP03}. The weight distribution of  $\cH_{(q,m)}$  is given in the following lemma
\cite{DL17}.

\begin{lemma}\label{lem-HCwt}
The weight distribution of $\cH_{(q,m)}$ is given by
\begin{eqnarray*}
q^m A_{k}(\cH_{(q,m)})=
& \sum_{\substack {0 \le i \le \frac{q^{m-1}-1}{q-1} \\0 \le j \le q^{m-1} \\ i+j=k}}\left[\binom{\frac{q^{m-1}-1}{q-1}}{i}
 \binom{q^{m-1}}{j}\Big((q-1)^k+(-1)^j(q-1)^i(q^m-1)\Big)\right]
\end{eqnarray*}
for $0 \leq k \leq (q^m-1)/(q-1)$.
\end{lemma}

The duals of the Hamming codes $\cH_{(q,m)}$ are called \emph{Simplex codes\index{Simplex code}}, which have
parameters $[(q^m-1)/(q-1), m, q^{m-1}]$.
The nonzero codewords of the $[(q^m-1)/(q-1), m, q^{m-1}]$ Simplex codes all have weight $q^{m-1}$.

\begin{theorem}\label{thm-Hamshort1}
Let $n=(q^m-1)/(q-1) \geq 4$, and let $t_1$ be any coordinator of codewords in $\cH_{(q,m)}$. Then the following hold:
\begin{itemize}
\item $(\cH_{(q,m)})_{\{t_1\}}$
is an $[n-1, n-m-1, 3]$ code over $\gf(q)$ with
\begin{eqnarray*}
A_{k}((\cH_{(q,m)})_{\{t_1\}}) = \frac{n-k}{n} A_k(\cH_{(q,m)})
\end{eqnarray*}
for $0 \leq k \leq n-1$, where $A_k(\cH_{(q,m)})$ was given in Lemma \ref{lem-HCwt}.
\item $(\cH_{(q,m)}^\perp)_{\{t_1\}}$
is an $[n-1, m-1, q^{m-1}]$ code over $\gf(q)$ with weight enumerator
$1+(q^{m-1}-1)z^{q^{m-1}}$.
\item $((\cH_{(q,m)})_{\{t_1\}})^\perp$ is an $[n-1, m, q^{m-1}-1]$ code over $\gf(q)$ with weight enumerator
$$
1+(q-1)q^{m-1} z^{q^{m-1}-1} +(q^{m-1}-1)z^{q^{m-1}}.
$$
\item $((\cH_{(q,m)}^\perp)_{\{t_1\}})^\perp$ is an $[n-1, n-m, 2]$ code over $\gf(q)$ with weight enumerator
\begin{eqnarray}\label{eqn-april401}
\frac{1}{q^{m-1}} [(1+(q-1)z)^{n-1} +(q^{m-1}-1)(1-z)^{q^{m-1}} ( 1+(q-1)z)^{n-1-q^{m-1}} ].
\end{eqnarray}
\end{itemize}
\end{theorem}

\begin{proof}
By Lemma \ref{lem-simpled}, $\cB_{q^{m-1}}(\cH_{(q,m)}^\perp)$ does not have repeated blocks. It is known that
the incidence structure
$(\cP(\cH_{(q,m)}^\perp), \cB_{q^{m-1}}(\cH_{(q,m)}^\perp))$ is a $2$-design \cite{DL17}. Since the Simplex code
$\cH_{(q,m)}^\perp$ has weight enumerator $1+(q^m-1)z^{q^{m-1}}$,
$(\cP(\cH_{(q,m)}^\perp), \cB_{k}(\cH_{(q,m)}^\perp))$ is the trivial $2$-design $(\cP(\cH_{(q,m)}^\perp), \{\emptyset\})$
or $(\cP(\cH_{(q,m)}^\perp), \emptyset)$
for each $k$ with $0 \leq k \leq n$ and $k \neq q^{m-1}$. It then follows from Theorem \ref{thm-tdesign-wtcode} that
$(\cP(\cH_{(q,m)}), \cB_{k}(\cH_{(q,m)}))$ is a $2$-design for each $k$ with $0 \leq k \leq n$. The desired conclusions
on  $(\cH_{(q,m)})_{\{t_1\}}$ and $(\cH_{(q,m)}^\perp)_{\{t_1\}}$  then follow from  Theorem \ref{thm-sctcode} and
Lemma  \ref{lem-HCwt}.

We now prove the conclusions on the code $((\cH_{(q,m)})_{\{t_1\}})^\perp$.  It follows from Theorems \ref{thm-shortenedcode}
and \ref{thm-pctcode} that
\begin{eqnarray}\label{eqn-april402}
A_k( ((\cH_{(q,m)})_{\{t_1\}})^\perp )
&=& A_k( ((\cH_{(q,m)})^\perp)^{\{t_1\}} ) \nonumber \\
&=& \sum_{i=0}^1 \frac{\binom{n-1}{k} \binom{k+i}{1} \binom{1}{i}}{\binom{n-1}{k-1+i} \binom{n}{1}} A_{k+i}((\cH_{(q,m)})^\perp).
\end{eqnarray}
Notice that $(\cH_{(q,m)})^\perp$ has weight enumerator $1+(q^m-1)z^{q^{m-1}}$. Combining this with (\ref{eqn-april402}),
we deduce that $A_k( ((\cH_{(q,m)})_{\{t_1\}})^\perp ) =0$ for all $k \not\in \{0, q^{m-1}-1, q^{m-1}\}$ and
$$
A_{q^{m-1}-1}( ((\cH_{(q,m)})_{\{t_1\}})^\perp )=(q-1)q^{m-1}
$$
and
$$
A_{q^{m-1}}( ((\cH_{(q,m)})_{\{t_1\}})^\perp )= q^{m-1}-1.
$$
This completes the proof of the desired conclusions on $((\cH_{(q,m)})_{\{t_1\}})^\perp$.

Finally, we prove the conclusions on the code $((\cH_{(q,m)}^\perp)_{\{t_1\}})^\perp$. Note that the weight enumerator
of $((\cH_{(q,m)}^\perp)_{\{t_1\}})^\perp$ is $1+(q^{m-1}-1)z^{q^{m-1}}$. The desired weight enumerator in (\ref{eqn-april401})
then follows from the MacWilliams identity. It is easily verified that the coefficient of $z$ in the polynomial in  (\ref{eqn-april401})
equals $0$, and the coefficient of $z^2$ is
\begin{eqnarray*}
\frac{ \binom{n-1}{2}(q-1)^2 + (q^{m-1}-1)[ \binom{q^{m-1}}{2} +\binom{n-1-q^{m-1}}{2} (q-1)^2 -(q-1)q^{m-1} (n-1-q^{m-1})   ]    }{q^{m-1}}  \\
= \frac{q(q-1)(q^{m-1}-1)}{2} >0.
\end{eqnarray*}
Consequently, $((\cH_{(q,m)}^\perp)_{\{t_1\}})^\perp$ has minimum weight $2$.
\end{proof}

\begin{table}[ht]
\begin{center}
\caption{Examples of the code $(\cH_{(q,m)})_{\{t_1\}}$}\label{tab-331-1}
\begin{tabular}{cccc} \hline
$q$  & $m$ & $[n, \kappa, d]$ &  Optimality   \\ \hline
$2$  & $3$  & $[6,3,3]$        & Yes \\
$2$  & $4$  & $[14,10,3]$        & Yes \\
$2$  & $5$  & $[30,25,3]$        & Yes \\
$2$  & $6$  & $[62,56,3]$        & Yes \\
$2$  & $7$  & $[126,119,3]$        & Yes \\
$3$  & $2$  & $[3, 1, 3]$        & Yes \\
$3$  & $3$  & $[12, 9, 3]$        & Yes \\
$3$  & $4$  & $[39, 35, 3]$        & Yes \\
$3$  & $5$  & $[120, 115, 3]$        & Yes \\
\hline
\end{tabular}
\end{center}
\end{table}

Theoretically, we have the following conclusions about the two shortened codes in Theorem \ref{thm-Hamshort1} and
their duals.
\begin{itemize}
\item Let $n \geq 7$ and $m \geq 2$. Then the shortened code $(\cH_{(q,m)})_{\{t_1\}}$ is both length-optimal
          and dimension-optimal with respect to the sphere-packing bound.
\item  The code $(\cH_{(q,m)}^\perp)_{\{t_1\}}$ meets the Griesmer bound.
\item  The code $((\cH_{(q,m)}^\perp)_{\{t_1\}})^\perp$ is distance-optimal with respect to the sphere-packing bound,
          and is MDS when $m=2$. Note that $((\cH_{(q,m)}^\perp)_{\{t_1\}})^\perp=(\cH_{(q,m)})^{\{t_1\}}$, which is a
          punctured Hamming code.
\item  The code $((\cH_{(q,m)})_{\{t_1\}})^\perp$  meets the Griesmer bound.
\end{itemize}

Table \ref{tab-331-1} lists  examples of the code $(\cH_{(q,m)})_{\{t_1\}}$, which show that the code is
distance-optimal in all these cases according to \cite{Macus}.  Table \ref{tab-331-2} lists  examples of the code $(\cH_{(q,m)}^\perp)_{\{t_1\}}$,
which show that the code is distance-optimal in all these cases according to \cite{Macus}.

\begin{table}[ht]
\begin{center}
\caption{Examples of the code $(\cH_{(q,m)}^\perp)_{\{t_1\}}$}\label{tab-331-2}
\begin{tabular}{cccc} \hline
$q$  & $m$ & $[n, \kappa, d]$ &  Optimality   \\ \hline
$2$  & $3$  & $[6,2,4]$        & Yes \\
$2$  & $4$  & $[14,3,8]$        & Yes \\
$2$  & $5$  & $[30, 4, 16]$        & Yes \\
$2$  & $6$  & $[62, 5, 32]$        & Yes \\
$2$  & $7$  & $[126, 6, 64]$        & Yes \\
$3$  & $2$  & $[3, 1, 3]$        & Yes \\
$3$  & $3$  & $[12, 2, 9]$        & Yes \\
$3$  & $4$  & $[39, 3, 27]$        & Yes \\
$3$  & $5$  & $[120, 4, 81]$        & Yes \\
\hline
\end{tabular}
\end{center}
\end{table}

\begin{theorem}\label{thm-Hamshort2}
Let $n=(q^m-1)/(q-1) \geq 6$, and let $t_1$ and $t_2$ be any two distinct coordinators of codewords in $\cH_{(q,m)}$. Then the following hold:
\begin{itemize}
\item $(\cH_{(q,m)})_{\{t_1,t_2\}}$
is an $[n-2, n-m-2, 3]$ code over $\gf(q)$ with
\begin{eqnarray*}
A_{k}((\cH_{(q,m)})_{\{t_1,t_2\}}) = \frac{\binom{k}{2} \binom{n-2}{k}}{ \binom{n}{2} \binom{n-2}{k-2}} A_k(\cH_{(q,m)})
\end{eqnarray*}
for $0 \leq k \leq n-2$, where $A_k(\cH_{(q,m)})$ was given in Lemma \ref{lem-HCwt}.
\item
$(\cH_{(q,m)}^\perp)_{\{t_1,t_2\}}$
is an $[n-2, m-2, q^{m-1}]$ code over $\gf(q)$ with weight enumerator
$$1+(q^{m-2}-1)z^{q^{m-1}}.$$
\item $((\cH_{(q,m)})_{\{t_1,t_2\}})^\perp$ is an $[n-2, m, q^{m-1}-2]$ code over $\gf(q)$ with weight enumerator
$$
1+ (q-1)^2q^{m-2} z^{q^{m-1}-2} + 2(q-1)q^{m-2} z^{q^{m-1}-1} +(q^{m-2}-1)z^{q^{m-1}}.
$$
\end{itemize}
\end{theorem}

\begin{proof}
The proof is similar to that of Theorem \ref{thm-Hamshort1}, and omitted here.
\end{proof}

\begin{table}[ht]
\begin{center}
\caption{Examples of the code $(\cH_{(q,m)})_{\{t_1,t_2\}}$}\label{tab-331-3}
\begin{tabular}{cccc} \hline
$q$  & $m$ & $[n, \kappa, d]$ &  Optimality   \\ \hline
$2$  & $3$  & $[5, 2, 3]$        & Yes  \\
$2$  & $4$  & $[13, 9, 3]$        & Yes \\
$2$  & $5$  & $[29, 24, 3]$        & Yes \\
$2$  & $6$  & $[61, 55, 3]$        & Yes \\
$2$  & $7$  & $[125, 118, 3]$        & Yes \\
$3$  & $3$  & $[11, 8, 3]$        & Yes \\
$3$  & $4$  & $[38, 34, 3]$        & Yes \\
$3$  & $5$  & $[119, 114, 3]$        & Yes \\
\hline
\end{tabular}
\end{center}
\end{table}

Let $n \geq 7$ and $m \geq 2$. Then the shortened code $(\cH_{(q,m)})_{\{t_1,t_2\}}$ is both length-optimal
          and dimension-optimal with respect to the sphere-packing bound. The code $((\cH_{(q,m)})_{\{t_1,t_2\}})^\perp$
          is MDS when $m=2$, and meets the Griesmer bound when $q >2$.
Table \ref{tab-331-3} lists  examples of the code $(\cH_{(q,m)})_{\{t_1,t_2\}}$, which show that the code is
distance-optimal in all these cases according to \cite{Macus}.  Table \ref{tab-331-4} lists  examples of the code $(\cH_{(q,m)}^\perp)_{\{t_1,t_2\}}$,
which show that the code is distance-optimal in most cases, and almost distance-optimal (i.e., the minimum dsitance
is one less than the best possible value) in three cases according to \cite{Macus}.

\begin{table}[ht]
\begin{center}
\caption{Examples of the code $(\cH_{(q,m)}^\perp)_{\{t_1,t_2\}}$}\label{tab-331-4}
\begin{tabular}{cccc} \hline
$q$  & $m$ & $[n, \kappa, d]$ &  Optimality   \\ \hline
$2$  & $3$  & $[5, 1, 4]$        & Almost \\
$2$  & $4$  & $[13, 2, 8]$        & Yes \\
$2$  & $5$  & $[29, 3, 16]$        & Yes \\
$2$  & $6$  & $[61, 4, 32]$        & Yes \\
$2$  & $7$  & $[125, 5, 64]$        & Yes \\
$3$  & $3$  & $[11, 1, 9]$        & Almost \\
$3$  & $4$  & $[38, 2, 27]$        & Almost \\
$3$  & $5$  & $[119, 3, 81]$        & Yes \\
\hline
\end{tabular}
\end{center}
\end{table}

\begin{theorem}\label{thm-Hamshort3}
Let $n=(q^m-1)/(q-1) \geq 6$, and let $T$ be the support of any codeword of weight 3 in $\cH_{(q,m)}$. Then
$(\cH_{(q,m)}^\perp)_{T}$ is an $[n-3, m-2, q^{m-1}]$ code over $\gf(q)$.
\end{theorem}

\begin{proof}
By Theorem \ref{thm-shortenedcode}, $(\cH_{(q,m)}^\perp)_{T}$ has dimension $m-2$. Since $m \geq 3$,
this code has $q^{m-2}  -1 \geq q-1 $ nonzero codewords. Consequently, $(\cH_{(q,m)}^\perp)_{T}$ has
minimum distance $q^{m-1}$.
\end{proof}

\begin{theorem}\label{thm-Hamshort4}
Let $n=(q^m-1)/(q-1)$, and let $T=\{i_1, i_2, \ldots, i_t\}$ be any subset of $t$ pairwise distinct coordinates in the codewords
in $\cH_{(q,m)}$.

If $t < q^{m-1}$, then $(\cH_{(q,m)})_{T}$ is an $[n-t, n-m-t, d((\cH_{(q,m)})_{T})]$ code over $\gf(q)$, where $d((\cH_{(q,m)})_{T}) \geq 3$.

If $t = q^{m-1}$ and $T$ is the support of a codeword of weight $q^{m-1}$ in $\cH_{(q,m)}^\perp$, then $(\cH_{(q,m)})_{T}$ is an $[n-t, n-m-t+1, d((\cH_{(q,m)})_{T})]$ code over $\gf(q)$, where $d((\cH_{(q,m)})_{T}) \geq 3$.
\end{theorem}

\begin{proof}
The desired conclusions on the length and dimension of $(\cH_{(q,m)})_{T}$ follow from Theorem \ref{thm-shortenedcode}.
By definition, $d((\cH_{(q,m)})_{T}) \geq 3$.
\end{proof}

We inform the reader that the code $(\cH_{(q,m)})_{T}$ in Theorem \ref{thm-Hamshort4} is both length-optimal and
dimension-optimal with respect to the sphere-packing bound when $n>7$, $m \geq 2$ and $t=|T|=3$. Hence,
Theorem \ref{thm-Hamshort4} does include a family of optimal shortened codes.

When $t \in \{1,2\}$, the parameters and the weight distribution of the code  $(\cH_{(q,m)})_{T}$ were described in
Theorems \ref{thm-Hamshort1} and \ref{thm-Hamshort2}. In many other cases,  $d((\cH_{(q,m)})_{T})=3$. Below
we present a general result.

A partial $k$-spread $\mathcal S$ of the projective space $\mathrm{PG}(m-1,q)$ is a collection of pairwise disjoint $k$-dimensional subspaces. By definition, every point of $\mathrm{PG}(m-1,q)$ is contained in at most one element of $\mathcal S$.

\begin{lemma}[\cite{Beut79}]\label{lem:t-cover}
Let $1\le k <m-1$, where $m$ and $k$ are integers, and let $r$ be the remainder  of $m$ divided by $k+1$.
Then there exists a partial $k$-spread with cardinality $\frac{q^m-q^r}{q^{k+1}-1} -q^r +1$.
\end{lemma}

\begin{theorem}
Let $n=\frac{q^m-1}{q-1}$, and let $T=\{i_1, i_2, \cdots, i_t\}$ be any subset of $t$ pairwise distinct coordinates in the
codewords in $\mathcal H_{(q,m)}$ with $t< \frac{q-1}{q+1}n$ if $m$ is even,  and $t< \frac{q-1}{q+1}(n-q^2)$ if $m$ is odd. Then the
minimum distance of the shortened Hamming code $\left ( \mathcal{H}_{(q,m)} \right )_{T}$ is equal to 3.
\end{theorem}

\begin{proof}
Let $P_0$, $\cdots$, $P_{n-1}$ be the points of $\mathrm{PG}(m-1,q)$ and let
$D=\{P_i: 0\le i \le n-1, i \not \in T\}$. By Lemma \ref{lem:t-cover}, we can choose
a $1$-spread $\mathcal S= \{U_1, \cdots, U_{L}\}$ with $L=\frac{q^m-q^r}{q^{2}-1} -q^r +1$, where
$$
r=\begin{cases}
0, \text{ if } m \text{ is even,}\\
1, \text{ otherwise.}
\end{cases}
$$
It is obvious that
$$
(q-1) L=\begin{cases}
\frac{q-1}{q+1}n, \text{ if } m \text{ is even,}\\
\frac{q-1}{q+1}(n-q^2), \text{ otherwise.}
\end{cases}
$$
By the pigeonhole principle, there must be a $j\in \{1, \cdots, L\}$
such that $|U_j \cap (\mathrm{PG}(m-1,q) \setminus D)| <q-1$. This clearly forces
$|U_j \cap D| \ge 3$. Then we choose three distinct points $P_{i_1'}, P_{i_2'}, P_{i_3'}$ in $U_j \cap D$.
Thus $P_{i_1'}, P_{i_2'}$ and $P_{i_3'}$ are collinear in $\mathrm{PG}(q,m-1)$.
It is easily seen that any nontrivial linear relationship among $P_{i_1'}, P_{i_2'}, P_{i_3'}$
gives rise to a codeword of weight $3$ in $\left ( \mathcal{H}_{(q,m)} \right )_{T}$.
This completes the proof.
\end{proof}

\section{Some shortened codes of the Reed-Muller codes}

Reed-Muller codes can be defined by either the univariate or the multivariate approach. Each approach has advantages
and disadvantages. We briefly recall the univariate
definition below. Let $m$ be a positive integer. Any function from $\gf(2^m)$ to $\gf(2)$ is called a (univariate)
Boolean function. Let $\bB_m$ denote the set of all Boolean functions on $\gf(2^m)$.  Every nonzero Boolean
function $f$ on $\gf(2^m)$ can be uniquely expressed as
\begin{eqnarray}\label{eqn-4601}
f(x)=\sum_{i=0}^{2^m-1} f_i x^i,
\end{eqnarray}
where $f_i \in \gf(2^m)$. Every integer $i$ with $0 \leq i \leq 2^m-1$ has the unique $2$-adic expansion
$i=\sum_{j=0}^{m-1} i_j 2^j$, where $i_j \in \{0,1\}$. The $2$-weight of $i$ is defined to be the Hamming weight
of $(i_0, i_1, \ldots, i_{m-1})$.  The algebraic degree of a Boolean function $f$ on $\gf(2^m)$ defined in (\ref{eqn-4601}),
denoted by $\deg(f)$, is defined to be the maximum 2-weight of all $i$ such that $f_i \neq 0$.

Let $\alpha$ be a generator of $\gf(2^m)^*$. Define $P_0=0$ and $P_i=\alpha^{i-1}$ for $1 \leq i \leq 2^m-1$.
The Reed-Muller code of length $2^m$ and order $r$ is defined by
$$
\cR(r,m)=\{(f(P_0), f(P_1), \ldots, f(P_{2^m-1})): f \in \bB_m, \deg(f) \leq r\}.
$$
The reader is referred to
\cite[Chapter 5]{DingBk2} and \cite{MS77} for detailed information on the Reed-Muller code.
The following are well known:
\begin{itemize}
\item  $\cR(r,m)^\perp=\cR(m-1-r, m)$.
\item $\cR(1,m)$ has parameters $[2^m, m+1, 2^{m-1}]$ and weight enumerator $1+(2^{m+1}-2)z^{2^{m-1}}+z^{2^m}$.
\item $\cR(m-2,m)$ has parameters $[2^m, 2^m-m-1, 4]$.
\end{itemize}

Our objective in this section is to study some shortened codes
of $\cR(1,m)$ and $\cR(m-2,m)$ and their duals. As before, we are only interested in optimal codes or codes meeting
a bound for linear codes.
The first result of this section is the following.

\begin{theorem}\label{THM-4305}
Let $m \geq 3$, and let $t_1$ be any coordinate. Then the following hold.
\begin{enumerate}
\item $\cR(1,m)_{\{t_1\}}$ is a $[2^m-1, m, 2^{m-1}]$ binary code with weight enumerator $1+(2^m-1)z^{2^{m-1}}$,
          and is equivalent to the binary Simplex code. The code meets the Griesmer bound.
\item  $(\cR(1,m)_{\{t_1\}})^\perp$ is a $[2^m-1,2^m- m-1, 3]$ binary code, and is equivalent to the binary Hamming code.
           The code meets the sphere-packing bound and is perfect.
\item    $\cR(m-2,m)_{\{t_1\}}$ is a $[2^m-1, 2^m-m-2, 4]$ binary code. The code is distance-optimal with respect to the
             sphere-packing bound.
\item   $(\cR(m-2,m)_{\{t_1\}})^\perp$ is a $[2^m-1, m+1, 2^{m-1}-1]$ binary code with weight enumerator
            $$
            1+(2^m-1)z^{2^{m-1}-1} +(2^m-1)z^{2^{m-1}} + z^{2^m-1}.
            $$
            This code meets the Griesmer bound.
\end{enumerate}
\end{theorem}

\begin{proof}
We outline the proof as follows. Note that $\cB_{2^m}(\cR(1,m))=\{\cP(\cR(1,m))\}$. It is then straightforward to see that
$(\cP(\cR(1,m)), \cB_{2^m}(\cR(1,m) ))$ is a $3$-$(2^m, 2^m, 1)$ simple design. It is known that
$(\cP(\cR(1,m)), \cB_{2^{m-1}}(\cR(1,m)))$ is a $3$-$(2^m, 2^{m-1}, 2^{m-2}-1)$ simple design \cite[p. 143]{DingBk2}. Clearly,
$(\cP(\cR(1,m)), \cB_{k}(\cR(1,m) ))$ is the trivial $3$-design $(\cP(\cR(1,m)), \{\emptyset\} )$ or $(\cP(\cR(1,m)), \emptyset )$ for all $k
\not\in \{2^{m-1}, 2^m\}$. It then follows from Theorem \ref{thm-tdesign-wtcode} that $(\cP(\cR(1,m)^\perp), \cB_{k}(\cR(1,m)^\perp ))$
is a $3$-design for all $k$ with $0 \leq k \leq 2^m$. We are now ready to apply Theorems \ref{thm-sctcode} and
\ref{thm-pctcode}. The desired conclusions on the codes can be similarly proved as the conclusions of Theorem
\ref{thm-Hamshort1}. The details are omitted here.
\end{proof}

Similarly, one can prove the following result.

\begin{theorem}\label{THM-4306}
Let $m \geq 3$, and let $t_1$ and $t_2$ be two distinct coordinates. Then the following hold.
\begin{enumerate}
\item $\cR(1,m)_{\{t_1, t_2\}}$ is a $[2^m-2, m-1, 2^{m-1}-1]$ binary code with weight enumerator $$1+(2^{m-1}-1)z^{2^{m-1}-1}.$$
         This code almost meets the Griesmer bound.
\item  $(\cR(1,m)_{\{t_1, t_2\}})^\perp$ is a $[2^m-2, 2^m-m-1, 2]$ binary code, and is distance-optimal with respect to the sphere-packing bound.
\item $\cR(m-2,m)_{\{t_1, t_2\}}$ is a $[2^m-2, 2^m-m-3, 4]$ binary code and is distance-optimal with respect to the sphere-packing bound.
\item $(\cR(m-2,m)_{\{t_1, t_2\}})^\perp$ is a $[2^m-2, m+1, 2^{m-1}-2]$ binary code with weight enumerator
$$
1+(2^{m-1}-1)z^{2^{m-1}-2} +2^m z^{2^{m-1}-1} +(2^{m-1}-1)z^{2^{m-1}} + z^{2^m-2}.
$$
This code almost meets the Griesmer bound.
\end{enumerate}
\end{theorem}

Note that all the codes in Theorems \ref{THM-4305} and \ref{THM-4306} are either optimal or almost optimal. We have also
the following.

\begin{theorem}\label{THM-4309}
Let $m \geq 3$, and let $t_1$, $t_2$ and $t_3$ be three pairwise distinct coordinates. Then the following hold.
\begin{enumerate}
\item $(\cR(m-2,m)_{\{t_1, t_2,t_3\}})^\perp$ is a $[2^m-3, m+1, 2^{m-1}-3]$ binary code with weight enumerator
$$
1 + (2^{m-2}-1)z^{2^{m-1}-3}  + 3\times 2^{m-2} z^{2^{m-1}-2} + 3\times 2^{m-2} z^{2^{m-1}-1} + (2^{m-2}-1)z^{2^{m-1}} + z^{2^m-3}.
$$
This code almost meets the Griesmer bound.
\item $\cR(m-2,m)_{\{t_1, t_2,t_3\}}$ is a $[2^m-3, 2^m-m-4, 4]$ binary code and is distance-optimal with respect to the sphere-packing bound.
\end{enumerate}
\end{theorem}

\begin{proof}
The desired conclusions can be similarly proved with Theorems \ref{thm-tdesign-wtcode}, \ref{thm-sctcode}, and \ref{thm-pctcode}.
The parameters of the two codes can also be settled with Theorems \ref{thm-april4101} and \ref{thm-april4102}. The details are omitted.
\end{proof}

\section{Some shortened codes of the ovoid codes}

In the projective space $\PG(3, \gf(q))$ with $q>2$, an \emph{ovoid}\index{ovoid} $\cV$
is a set of $q^2+1$ points such that no three of them are collinear (i.e., on the same line). In other words, an ovoid is a $(q^2+1)$-cap (a cap with $q^2+1$ points) in
$\PG(3, \gf(q))$, and thus a maximal cap. Two ovoids
are said to be \emph{equivalent}\index{equivalent ovoids} if there is a collineation (i.e.,
automorphism) of $\PG(3, \gf(q))$ that sends one to the other.

A \emph{classical ovoid} $\cV$ can be defined as the set of all points given by
\begin{eqnarray}
\cV=\{(0,0,1, 0)\} \cup \{(x,\, y,\, x^2+xy +ay^2,\, 1): x,\, y \in \gf(q)\},
\end{eqnarray}
where $a \in \gf(q)$ is such that the polynomial $x^2+x+a$ has no root in $\gf(q)$.
Such ovoid is called an \emph{elliptic quadric}\index{quadric}, as the points
come from a non-degenerate elliptic quadratic form.

For $q=2^{2e+1}$ with $e \geq 1$, there is an ovoid which is not an elliptic quadric,
and is called the \emph{Tits oviod}\index{Tits ovoid}. It is defined by
\begin{eqnarray}
\cT=\{(0,0,1,0)\}\cup \{(x,\,y,\, x^{\sigma} + xy +y^{\sigma+2},\,1): x, \, y \in \gf(q)\},
\end{eqnarray}
where $\sigma=2^{e+1}$.

For odd $q$, any ovoid is an elliptic quadric. For even $q$, Tits ovoids are the only known
ones which are not elliptic quadratics. In the case that $q$ is even, the elliptic quadrics and
the Tits ovoid are not equivalent.

Let $\cV$ be an ovoid in $\PG(3, \gf(q))$ with $q>2$. Denote by
$$
\cV=\{\bv_1, \bv_2, \ldots, \bv_{q^{2}+1}\},
$$
where each $\bv_i$ is a column vector in $\gf(q)^4$. Let $\C_{\cV}$ be the linear code over $\gf(q)$
with generator matrix
\begin{eqnarray}
G_{\cV}=\left[\bv_1 \bv_2 \cdots \bv_{q^{2}+1}\right].
\end{eqnarray}
It is known that $\C_{\cV}$ is a $[q^2+1, 4, q^2-q]$ code over $\gf(q)$ with weight enumerator
\begin{eqnarray}\label{eqn-wtdistovidcode}
1+(q^2-q)(q^2+1)z^{q^2-q}+(q-1)(q^2+1)z^{q^2}
\end{eqnarray}
and its dual $\C_{\cV}^\perp$ is a $[q^2+1, q^2-3, 4]$ code over $\gf(q)$ \cite[Chapter 13]{DingBk2}.
Conversely, the set of column vectors of a generator matrix of any  $[q^2+1, q^2-3, 4]$ code over $\gf(q)$
is an ovoid in $\PG(3, \gf(q))$. Hence,  ovoids in $\PG(3, \gf(q))$ and $[q^2+1, q^2-3, 4]$ codes over $\gf(q)$
are the same, and a $[q^2+1, 4, q^2-q]$ code over $\gf(q)$ is called an ovoid code over $\gf(q)$.

The weight distribution of an ovoid  is given in the following lemma \cite[p. 324]{DingBk2}, and will be employed later.

\begin{lemma}\label{lem-iccodedual}
Let $q \geq 4$, and let $\C$ be a $[q^2+1, 4, q^2-q]$ code over $\gf(q)$.
Then the weight distribution of $\C^\perp$ is given by
\begin{eqnarray}
q^4 A_\ell(\C^\perp) &=& \binom{q^2+1}{\ell}(q-1)^\ell +
     u \sum_{i+j=\ell} \binom{q^2-q}{i}(-1)^i \binom{q+1}{j} (q-1)^j  + \nonumber \\
   && v\left[ (-1)^\ell \binom{q^2}{\ell} +(-1)^{\ell-1} (q-1) \binom{q^2}{\ell-1} \right]
\end{eqnarray}
for all $4 \leq \ell \leq q^2$, and
$$
q^4 A^\perp_{q^2+1}=(q-1)^{q^2+1}+u(q-1)^{q+1}+v(q-1),
$$
where
\begin{eqnarray}
u=(q^2-q)(q^2+1), \ v=(q-1)(q^2+1).
\end{eqnarray}
\end{lemma}

We are now ready to study some shortened codes of ovoid codes, and have the following results.

\begin{theorem}\label{thm-4301}
Let $q \geq 4$, and let $\C$ be a $[q^2+1, 4, q^2-q]$ code over $\gf(q)$. For any coordinate $t_1$,
the following hold.
\begin{enumerate}
\item $\C_{\{t_1\}}$ is a $[q^2, 3, q^2-q]$ code over $\gf(q)$ with weight enumerator
          \begin{eqnarray}\label{eqn-4201}
          1+q(q^2-1)z^{q^2-q} + (q-1)z^{q^2}.
          \end{eqnarray}
\item  $(\C_{\{t_1\}})^\perp$ is a $[q^2, q^2-3, 3]$ almost MDS code over $\gf(q)$.

\item $((\C^\perp)_{\{t_1\}})^\perp$ is a $[q^2, 4, q^2-q-1]$  code over $\gf(q)$ with weight enumerator
         $$
         1+q^2(q-1)^2z^{q^2-q-1} + q(q^2-1)z^{q^2-q} +q^2(q-1)z^{q^2-1} + (q-1)z^{q^2}.
         $$

\item $(\C^\perp)_{\{t_1\}}$ is a $[q^2, q^2-4, 4]$ almost MDS code over $\gf(q)$ with weight distribution
          $$
          A_k((\C^\perp)_{\{t_1\}})=\frac{\binom{k}{1} \binom{q^2}{k}}{\binom{q^2+1}{1} \binom{q^2}{k-1}} A_k(\C^\perp),
          $$
          where $1 \leq k \leq q^2$, and $A_k(\C^\perp)$ was given in Lemma \ref{lem-iccodedual}.
\end{enumerate}
\end{theorem}

\begin{proof}
It is straightforward to see that two codewords of weight $q^2$ in $\C$ have the same support if and only if one is a nonzero
multiple of the other. Hence, $\cB_{q^2}(\C)$ has no repeated block and cardinality $q^2+1$, and every $q^2$-subset of
$\cP(\C)$ appears exactly once in $\cB_{q^2}(\C)$. Consequently, $(\cP(\C), \cB_{q^2}(\C))$ is a $3$-$(q^2+1, q^2, q^2-2)$
design. It is known that  $(\cP(\C), \cB_{q^2-q}(\C))$ is a $3$-$(q^2+1, q^2-q, (q-2)(q^2-q-1))$ simple design \cite[p. 327]{DingBk2}.
Hence, $(\cP(\C), \cB_{k}(\C))$ is a $3$-design for all $k$ with $0 \leq k \leq q^2+1$. It then follows from Theorem
\ref{thm-tdesign-wtcode} that  $(\cP(\C^\perp), \cB_{k}(\C^\perp))$ is a $3$-design for  all $k$ with $0 \leq k \leq q^2+1$.

The desired conclusions on $\C_{\{t_1\}}$ then follow from Theorem \ref{thm-sctcode} and the weight enumerator of
$\C$ given in (\ref{eqn-wtdistovidcode}). Using the weight enumerator of $\C_{\{t_1\}}$  in (\ref{eqn-4201}) and the
MacWilliams identity, one can prove that the minimum distance $d((\C_{\{t_1\}})^\perp)=3$. Thus, $(\C_{\{t_1\}})^\perp$ is a $[q^2, q^2-3, 3]$ almost MDS code over $\gf(q)$.

We now prove the desired conclusions on  $((\C^\perp)_{\{t_1\}})^\perp$. By Theorem \ref{thm-shortenedcode}, we have
$$
((\C^\perp)_{\{t_1\}})^\perp =\C^{\{t_1\}}.
$$
It then follows from Theorem \ref{thm-pctcode} that
\begin{eqnarray}\label{eqn-4303}
A_k(((\C^\perp)_{\{t_1\}})^\perp) = A_k(\C^{\{t_1\}}) =
 \sum_{i=0}^1 \frac{\binom{q^2}{k} \binom{k+i}{1}}{\binom{q^2}{k-1+i} \binom{q^2+1}{1}}  A_{k+i} (\C).
\end{eqnarray}
The desired conclusions then follow from the weight enumerator of $\C$ in (\ref{eqn-wtdistovidcode}).

Since $d(\C^\perp)=4$, by definition $d((\C^\perp)_{\{t_1\}}) \geq 4$. If  $d((\C^\perp)_{\{t_1\}}) = 5$, then
$(\C^\perp)_{\{t_1\}}$ would be a $[q^2, q^2-4, 5]$ MDS code, and $((\C^\perp)_{\{t_1\}})^\perp$ would be
$[q^2, 4, q^2-3]$ MDS code, which is a contradiction. Therefore,  $d((\C^\perp)_{\{t_1\}}) = 4$. The desired
conclusion on the weight distribution of $(\C^\perp)_{\{t_1\}}$ then follows from (\ref{eqn-4303}) and the weight enumerator of
$\C$ given in (\ref{eqn-wtdistovidcode}).
\end{proof}

Notice that all ovoid codes meet the Griemer bound. The shortened codes and their duals documented in
Theorem \ref{thm-4301} are very interesting due to the following.
\begin{itemize}
\item All the four classes of codes in Theorem \ref{thm-4301} support $2$-designs.
\item Both $\C_{\{t_1\}}$ and $(\C^\perp)_{\{t_1\}}$ meet the Griesmer bound.
\end{itemize}

Using the Assmus-Mattson theorem (see \cite{AM69} or \cite[Chapter 4]{DingBk2}), Theorem \ref{thm-4301} and Lemma \ref{lem-simpled}, one can prove the following theorem.
We omit the proofs here.

\begin{theorem}\label{thm-4302}
Let $q \geq 4$, and let $\C$ be a $[q^2+1, 4, q^2-q]$ code over $\gf(q)$. For any coordinate $t_1$,
the following hold.
\begin{enumerate}
\item  The incidence structure $(\cP(\C_{\{t_1\}}), \cB_{q^2-q}(\C_{\{t_1\}}) )$  is a $2$-$(q^2, q^2-q, q^2-q-1)$ simple design.

\item The incidence structure $(\cP((\C_{\{t_1\}})^\perp), \cB_{q^2-q}((\C_{\{t_1\}})^\perp) )$  is a $2$-$(q^2, 3, \lambda)$ simple design for some integer $\lambda$.

\item The incidence structure $(\cP(((\C^\perp)_{\{t_1\}})^\perp),  \cB_{q^2-q-1}(((\C^\perp)_{\{t_1\}})^\perp) )$
         is a $2$-$(q^2, q^2-q-1, (q-2)(q^2-q-1))$ simple design. The complement of this design is a $2$-$(q^2,q+1,q)$ design.

\item   The incidence structure $(\cP(((\C^\perp)_{\{t_1\}})^\perp),  \cB_{q^2-q}(((\C^\perp)_{\{t_1\}})^\perp) )$
         is a $2$-$(q^2, q^2-q, q^2-q-1)$ simple design.  The complement of this design is a Steiner system
         $2$-$(q^2, q, 1)$, i.e., an affine plane.

\item  The incidence structure $(\cP((\C^\perp)_{\{t_1\}}),  \cB_{4}((\C^\perp)_{\{t_1\}}))$ is a $2$-$(q^2, 4, \lambda)$ simple design
          for some integer $\lambda$.
\end{enumerate}
\end{theorem}

Since $(\cP(\C^\perp), \cB_{k}(\C^\perp))$ is a $3$-design for  all $k$ with $0 \leq k \leq q^2+1$, we can similarly determine
the parameters and weight distributions of the codes $(\C^\perp)_{\{t_1, t_2\}}$ and  $(\C^\perp)_{\{t_1, t_2,t_3\}}$. However,
these shortened codes are less interesting.

\section{Summary and concluding remarks}

The main contributions of this paper are the following.
\begin{itemize}
\item It was proved in Theorem \ref{thm-maingen1} that every linear code $\C$ over $\gf(q)$ with minimum distance $d \geq 2$ is a shortened code of $\C(m,q,\alpha)^\perp$
for some $m$, $q$ and $\alpha$, where $\alpha$ is a generator of $\gf(q)^*$ and $\C(m,q,\alpha)$ was defined in
(\ref{eqn-irredcycliccode}). This showed the importance of the shortening technique and the family of cyclic codes
$\C(m,q,\alpha)^\perp$.
\item It was proved in Theorem \ref{thm-maingen2} that every linear code over $\gf(q)$ with minimum distance at least $3$ is a shortened code of a Hamming code over $\gf(q)$.  This showed the importance of the shortening technique and the family of Hamming codes.
\item The parameters and weight distributions of two families of shortened codes of the Hamming codes and their duals were settled in Theorem
\ref{thm-Hamshort1}. All of the shortened codes are optimal.
\item The parameters and weight distributions of another two families of shortened codes of the Hamming codes and their duals were settled in Theorem
\ref{thm-Hamshort2}. All of the shortened codes are optimal.
\item The parameters of three families of shortened codes of the Reed-Muller codes $\cR(1,m)$ and  $\cR(m-2,m)$ were settled in Theorems
\ref{THM-4305} and \ref{THM-4309}. All of the shortened codes are optimal.
\item The parameters of another two families of shortened codes of  the Reed-Muller codes $\cR(1,m)$ and  $\cR(m-2,m)$ were settled in Theorem
\ref{THM-4306}. All of the shortened codes are either optimal or almost optimal.
\item The parameters of shortened codes of the oviod code and its dual were settled in Theorem \ref{thm-4301}. The shortened codes are either optimal or almost optimal.
\item Five families of $2$-designs were obtained from the shortened codes of the ovoid codes and their duals and were documented in Theorem \ref{thm-4302}.  Some of the $2$-designs are interesting.
\end{itemize}
In summary, eleven infinite families of optimal shortened codes with new parameters were presented in this paper.

Since every linear code with minimum weight at least $3$ is a shortened code of a Hamming code, some shortened codes must have bad parameters and some shortened codes must have good or optimal parameters. To obtain an optimal or good shortened code, a linear code $\C$ and the coordinate set $T$ for shortening must be properly selected.



\begin{thebibliography}{99}


\bibitem{AM69}
E. F. Assmus, Jr., H. F.  Mattson, Jr.,
``New 5-designs," \emph{J. Comb. Theory Ser. A}, vol. 6, no. 2, pp. 122--151, March 1969.

\bibitem{Beut79}
A. Beutelspacher, ``On $t$-covers in finite projective spaces,"
J. Geom., vol. 12, no. 1, pp, 10--16, 1979.




\bibitem{Chen72}
C. L. Chen, ``On shortened finite geometry codes," \emph{Information and Control}, vol. 20, pp. 216--221, 1972.


\bibitem{DingBk2}
C. Ding, \emph{Designs from Linear Codes}, World Scientific, Singapore, 2018.

\bibitem{DL17}
C. Ding, C. Li, ``Infinite families of 2-designs and 3-designs from linear codes,"
\emph{Discrete Math.}, vol. 340, no. 10, pp. 2415--2431, Oct. 2017.

\bibitem{DY13}
C. Ding, J. Yang, ``Hamming weights in irreducible cyclic codes," \emph{Disc. Math.},  vol. 313, no. 4,
pp. 434--446, April 2013.

\bibitem{Goldwasser}
J. L. Goldwasser, ``Shortened and punctured codes and the MacWilliams identity,"
\emph{Linear Algebra and Its Applications}, vol. 253, 1--13, 1997.

\bibitem{Macus}
M. Grassl, Code tables: bounds on the parameters of various types of codes,
http://www.codetables.de.

\bibitem{HS73}
H. J. Helgert, R. D. Stinaff, ``Shortened BCH codes," \emph{IEEE Trans. Inf. Theory}, vol. 19, no. 6,
pp. 818--820, 1973.

\bibitem{HengDing19}
Z. Heng, C. Ding, ``A construction of $q$-ary linear codes with irreducible cyclic codes,"
\emph{Des. Codes Cryptogr.}, vol. 87, pp. 1087--1108, 2019.

\bibitem{HWW20}
Z. Heng, W. Wang, Y. Wang, ``Projective binary linear codes from special Boolean functions,"
\emph{Appl. Algebra Eng. Commun. Comput.},  https://doi.org/10.1007/s00200-019-00412-z

\bibitem{Hsu71}
H. T. Hsu, ``A class of binary shortened cyclic codes for a compound channel,"
\emph{Information and Control}, vol. 18, pp. 126--139, 1971.



\bibitem{HP03}
W. C. Huffman, V. Pless, \emph{Fundamentals of
Error-Correcting Codes}, Cambridge University Press, Cambridge, 2003.

\bibitem{Kasa69}
T. Kasami, ``Optimum shortened cyclic codes for burst-error correction," \emph{IEEE Trans. Inf. Theory}, vol. 9, no. 2,
105--109, March 1963.


\bibitem{Lin72}
S. Lin, ``Shortened finite geometry codes," \emph{IEEE Trans. Inf. Theory}, vol. 18, no. 5, pp. 692--696, Sept. 1972.

\bibitem{MS77}
F. J. MacWilliams, N. J. A. Sloane, \emph{The Theory of Error-Correcting Codes},
North-Holland, Amsterdam,  1977.

\bibitem{NZ15}
P. Nelson, S. H. M. van Zwam, ``On the existence of asymptotically good linear codes in minor-closed class,"
\emph{IEEE Trans. Inf. Theory}, vol. 61, no. 3, pp. 1153--1158, March 2015.


\bibitem{TDX19}
C. Tang, C. Ding, M. Xiong, ``Steiner systems $S(2, 4, \frac{3^m-1}{2})$ and $2$-designs from
ternary linear codes of length $\frac{3^m-1}{2}$,"  \emph{Des. Codes Cryptogr.}, vol. 87, no. 12,
         pp. 2793--2811, December 2019.

\bibitem{TDX192}
C. Tang, C. Ding, M. Xiong, ``Codes, differentially $\delta$-uniform functions and $t$-designs,"
\emph{IEEE Trans. Inf. Theory}, vol. 66, no. 6, pp. 3691--3703, June 2020.


\bibitem{YP}
A. Yardi, R. Pellikaan, ``On shortened and punctured cyclic codes," arXiv:1705.09859v1 [cs.IT].

\end{thebibliography}
\end{document}